\documentclass[pdflatex,sn-mathphys-num]{sn-jnl}

\usepackage{hyperref}
\usepackage{url}

\usepackage{booktabs}

\usepackage{amsmath,amssymb,amsthm}
\usepackage{algorithm}
\usepackage{algpseudocode}
\usepackage{multirow}
\usepackage{wrapfig}
\usepackage{graphicx}

\usepackage{tikz-cd}
\usetikzlibrary{patterns}

\usepackage{placeins}

\theoremstyle{definition}
\newtheorem{definition}{Definition}[section]
\newtheorem{assumption}[definition]{Assumption}

\theoremstyle{plain}
\newtheorem{theorem}[definition]{Theorem}
\newtheorem{proposition}[definition]{Proposition}

\newtheorem{lemma}[definition]{Lemma}

\theoremstyle{remark}
\newtheorem{remark}[definition]{Remark}

\renewcommand{\orcidlogo}{%
  \includegraphics[width=10pt]{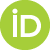}%
}

\begin{document}

\title[Higher-Order Persistence Diagrams]{Higher-Order Persistence Diagrams}

\author*{\fnm{Charles} \sur{Fanning} \orcid{https://orcid.org/0009-0002-8251-2802}}\email{cfannin8@students.kennesaw.edu}

\author{\fnm{Mehmet} \sur{Aktas} \orcid{https://orcid.org/0000-0002-9527-9600}}\email{maktas1@kennesaw.edu}

\affil{\orgdiv{School of Data Science and Analytics}, 
\orgname{Kennesaw State University}, 
\orgaddress{\street{1000 Chastain Rd NW}, 
\city{Kennesaw}, 
\postcode{30144}, 
\state{Georgia}, 
\country{United States}}}







\abstract{
Persistent homology tracks when topological features appear and disappear across filtered data and represents the lifetime of each feature by a persistence interval. In multiple instance learning, when the instances in a bag use filtrations over the same ordered parameter space, we can compare relationships between persistence intervals across instances. No existing method is designed to aggregate collections of persistence diagrams in multiple instance learning into a single persistence diagram of implications between intervals. We introduce higher-order persistence diagrams to represent relationships between persistence intervals and recursively construct each higher-order persistence interval from a pair of lower-order persistence intervals. We then aggregate these relationships across the persistence diagrams in a bag. Explicit construction of this aggregate is computationally expensive, so we introduce harmonic aggregation to compute coordinates of the aggregate from the lower-order persistence diagrams without explicitly constructing it. Computing a harmonic aggregation coordinate has lower asymptotic computational complexity than explicitly constructing the aggregate. We evaluate harmonic aggregation against persistence landscapes, persistence images, and PersLay on random network model and real-world graph-classification benchmarks, and find that harmonic aggregation has the highest accuracy on the random network model benchmark and on most real-world data benchmarks. Source code is available at \url{https://github.com/cfanning8/higher-order-PDs}.
}



\maketitle

\section{Introduction}\label{sec1}
%
%
%

In multiple instance learning (MIL), a label is assigned to a collection of observations rather than to each observation individually; this collection is called a \textit{bag}, whose observations are \textit{instances} \cite{DIETTERICH199731,carbonneau2018multiple,ilse2018attention}. While instances in MIL may generally take many different forms, we are interested in the case in which each instance is a persistence diagram. To represent each instance, we use persistent homology to compute a persistence diagram that records the birth and death parameters of topological features \cite{892133, Zomorodian2005ComputingPH, Oudot2015PersistenceT}. In this setting, each bag contains one persistence diagram for each instance. When we filter every instance in a bag over the same space of parameters, the persistence intervals lie in one common space, so the same relationship between persistence intervals has the same interpretation across instances.

Many MIL models aggregate vector-valued representations of instances \cite{carbonneau2018multiple, ilse2018attention}, and persistence-diagram vectorizations provide a direct way to apply these models to instances represented as persistence diagrams. \textit{Persistence landscapes} (PL), \textit{persistence images} (PI), and \textit{PersLay} are standard persistence-diagram vectorizations \cite{bubenik2015statistical, adams2017persistence, carriere2020perslay, pun2022survey}. Existing MIL methods that use persistent homology vectorize or scalarize the persistence diagrams before or during aggregation \cite{borowa2021deep,kazeminia2026topologicalinductivebiasfosters,obeid2025pmil,obeid2025prostate,kazeminia2026topomil,beuria2025probing}. To our knowledge, no existing MIL method directly aggregates persistence diagrams without ever converting the diagrams into vector- or scalar-valued representations.

In this paper we aggregate persistence diagrams into a higher-order persistence diagram without converting them into vector- or scalar-valued representations. We introduce higher-order persistence diagrams to represent relationships between lower-order persistence intervals. We construct each higher-order persistence interval from a pair of lower-order persistence intervals and apply the same construction recursively to define persistence diagrams of increasing order. Unlike established approaches from the literature, our aggregation retains information about which intervals co-occurred within each instance after aggregating across instances. Explicitly computing the aggregate requires quadratic time at each recursive step since we must consider pairs of intervals from the preceding order. To address this cost, we define \textit{harmonic aggregation} (HA), which uses Fourier analysis to map the aggregate into frequency space directly from the input diagrams without explicitly constructing the aggregate.

We evaluate our proposed harmonic aggregation for graph classification on random network models and real-world graph classification datasets. For PL, PI, and PersLay, we aggregate each bag by averaging the vectorized persistence diagrams. Harmonic aggregation outperforms PI, PL, and PersLay on the random network model benchmark and on most real-world data benchmarks.

This paper makes the following contributions:
\begin{itemize}
\item We introduce higher-order persistence diagrams by recursively constructing higher-order persistence intervals from pairs of lower-order persistence intervals.
\item We prove an exact factorization that computes a harmonic aggregation coordinate directly from lower-order input virtual persistence diagrams without explicitly constructing the aggregate. For fixed order dimension \(r\), this requires \(O(c2^N\log^{r-1}(c2^N))\) time, compared with \(O(c^2 4^N)\) time for explicit construction of the aggregate.
\item We show that harmonic aggregation achieves competitive graph-classification performance compared with PL, PI, and PersLay.
\end{itemize}

\section{Preliminaries}
\label{sec:background}

Interpretable aggregations rely on features with consistent meanings across observations. We consider a filtration \(K^0\subseteq K^1\subseteq\cdots\subseteq K^m\) and use persistent homology to track homology classes across the induced maps between the homology groups \(H_k(K^i)\) \cite{892133,Zomorodian2005ComputingPH,Oudot2015PersistenceT}. The birth and death parameters of these homology classes define persistence intervals, and the collection of these intervals forms a persistence diagram. When we filter every observation over the same ordered parameter space, the resulting persistence intervals lie in a common interval space. A given persistence interval, therefore, has the same interpretation across observations.

Statistical and machine learning methods commonly represent persistence diagrams as vectors through persistence diagram vectorizations \cite{ali2023survey,hensel2021survey}. Common examples include persistence landscapes and persistence images \cite{bubenik2015statistical,adams2017persistence}.

Classical persistence diagrams have nonnegative multiplicities, whereas some generalizations of persistence diagrams may have signed multiplicities related to M\"obius inversion of the rank function \cite{patel2018generalized, kimmemoli2021generalized, betthauser2022graded}. 

\begin{enumerate}
\item The rank function assigns to each pair of filtration values the number of homology classes that survive between them.
\item M\"obius inversion of the rank function gives the multiplicity of each persistence interval.
\end{enumerate}

To work with these signed multiplicities, we need an algebraic extension of persistence diagrams that admits subtraction.

Virtual persistence diagrams provide this representation. This construction begins with the commutative monoid of persistence diagrams over a metric pair \cite[Sections~2.2--2.3]{bubenik2022virtual}. The Grothendieck completion of this monoid is the free abelian group of virtual persistence diagrams, where persistence diagrams can be subtracted \cite[Section~2.5]{bubenik2022virtual}. Figure~\ref{fig:vpd-intro} illustrates the virtual persistence diagram obtained by subtracting two classical persistence diagrams.

\begin{figure}[t]
\centering
\includegraphics[width=\linewidth]{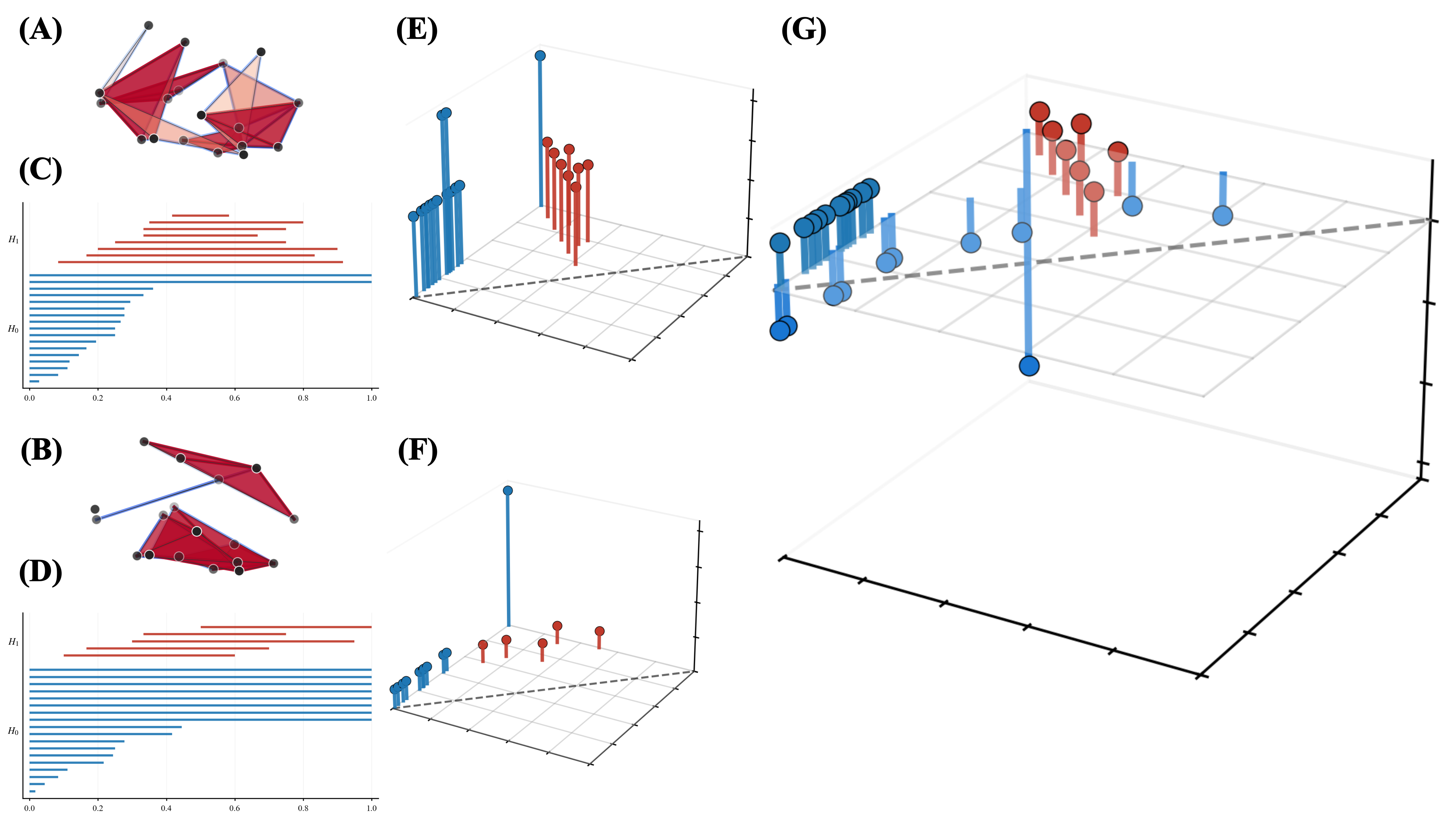}
\caption{
(A) A weighted simplicial complex defining the first filtration.
(B) A weighted simplicial complex defining the second filtration.
(C) The barcode computed from the filtration in (A).
(D) The barcode computed from the filtration in (B).
(E) The persistence diagram with birth--death multiplicities encoded by (C).
(F) The persistence diagram with birth--death multiplicities encoded by (D).
(G) The virtual persistence diagram \((E)-(F)\in K(X,A)\), whose signed multiplicity at each birth--death point equals the multiplicity in (E) minus the multiplicity in (F).
}
\label{fig:vpd-intro}
\end{figure}

\section{Higher-order persistence diagrams}

In multiple instance learning (MIL), a prediction is computed for a bag of observations, and aggregation methods combine these observations into a representation for prediction \cite{carbonneau2018multiple,ilse2018attention}. For interpretable aggregation, this representation should retain relationships among the features recorded for the individual observations.

Persistence intervals provide suitable features because their meanings remain consistent across observations. When the observations are filtered over the same ordered parameter space, their persistence intervals lie in a common filtration parameter space. A given relationship between persistence intervals, therefore, has the same interpretation across observations.

We represent these relationships by recursively pairing persistence intervals to form higher-order interval spaces and higher-order persistence diagrams.

\begin{assumption}\label{ass:base-structure}
Assume that \((X,d)\) is a metric space with a preorder \(\preceq\) of finite order dimension whose associated equivalence relation is closed in \(X\times X\). That is, there exist maps \(\phi_1,\ldots,\phi_r:X\to\mathbb R\) such that \(x\preceq y\) if and only if \(\phi_k(x)\leq\phi_k(y)\) for all \(1\leq k\leq r\).
\end{assumption}

We use Assumption~\ref{ass:base-structure} as the base of our proposed recursive hierarchy. Set \(X^{(0)}:=X\), \(d^{(0)}:=d\), and \(\preceq^{(0)}:=\preceq\). We call an element of \(X^{(n)}\) an order-\(n\) interval. Fix \(p\in[1,\infty]\). The following definition recursively constructs the higher-order interval spaces and their associated preorders, diagonals, and metrics.

\begin{definition}[Higher-order persistence intervals]\label{def:higher-order-persistence-intervals}
We define the order-\((n+1)\) interval space, preorder, diagonal, and metric as follows.
\begin{enumerate}
\item Define
\(
X^{(1)}
:=
X^{(0)}\times X^{(0)},
\)
and, for \(n\geq 1\), define
\[
X^{(n+1)}
:=
\bigl(X^{(n)}\setminus A^{(n)}\bigr)
\times
\bigl(X^{(n)}\setminus A^{(n)}\bigr).
\]

\item Define
\(
(a,b)\preceq^{(n+1)}(c,d)
\quad\Longleftrightarrow\quad
c\preceq^{(n)}a
\text{ and }
b\preceq^{(n)}d.
\)

\item Write \(a\sim^{(n)}b\) when
\(a\preceq^{(n)}b\) and \(b\preceq^{(n)}a\), and define
\[
A^{(n+1)}
:=
\{(a,b)\in X^{(n+1)}:a\sim^{(n)}b\}.
\]

\item For \(u=(a,b)\in X^{(n+1)}\), define
\[
d^{(n+1)}\bigl(u,A^{(n+1)}\bigr)
:=
\inf_{(x,y)\in A^{(n+1)}}
\left\|
\bigl(
d^{(n)}(a,x),
d^{(n)}(b,y)
\bigr)
\right\|_p.
\]

\item For \(u=(a,b)\) and \(v=(c,d)\) in \(X^{(n+1)}\), define
\[
d^{(n+1)}(u,v)
:=
\min\left\{
\begin{aligned}
\left\|
\bigl(
d^{(n)}(a,c),
d^{(n)}(b,d)
\bigr)
\right\|_p,
&d^{(n+1)}\bigl(u,A^{(n+1)}\bigr)+\\
&d^{(n+1)}\bigl(v,A^{(n+1)}\bigr)
\end{aligned}
\right\}.
\]
\end{enumerate}
\end{definition}

Persistence diagrams are finite multisets of persistence intervals. For a set \(X\), let \(D(X)\) denote the free commutative monoid on \(X\), whose elements are finite multisets of elements of \(X\). Following \cite[Sections~2.2--2.3]{bubenik2022virtual}, for \(A\subseteq X\), let \(D(A)\) denote the submonoid generated by \(A\). Define \( D(X,A) := D(X)/D(A). \) This quotient is naturally isomorphic to \(D(X\setminus A)\). Applying this quotient to each pair \(\bigl(X^{(n)},A^{(n)}\bigr)\) gives the hierarchy of higher-order persistence diagram monoids.

\begin{definition}[Higher-order persistence diagram monoids]\label{def:pd-monoid}
For \(n\geq1\), define the order-\(n\) persistence diagram monoid by
\[
D^{(n)}(X,A)
:=
D\bigl(X^{(n)},A^{(n)}\bigr).
\]
\end{definition}

\begin{lemma}\label{lem:higher-order-metric}
For every \(n\geq1\), \(d^{(n)}\) restricts to a metric on
\(X^{(n)}\setminus A^{(n)}\).
\end{lemma}

\usetikzlibrary{calc,patterns}

\begin{wrapfigure}{r}{0.43\textwidth}
\centering
\resizebox{\linewidth}{!}{%
\begin{tikzpicture}[
  x=3.65cm,
  y=3.65cm,
  every node/.style={font=\small},
  guide/.style={
    gray!55,
    densely dashed,
    line width=0.45pt
  },
  orderone/.style={
    circle,
    fill=black,
    draw=black,
    inner sep=0pt,
    minimum size=4.8pt
  },
  orderthree/.style={
    draw=blue!75!black,
    draw opacity=0.90,
    fill=blue!42,
    fill opacity=0.26,
    line width=1.2pt
  }
]

\newcommand{\ordertwocapsule}[2]{%
  \path
    let
      \p1=(#1),
      \p2=(#2),
      \n1={veclen(\x2-\x1,\y2-\y1)},
      \n2={atan2(\y2-\y1,\x2-\x1)}
    in
    node[
      draw=blue!55!cyan,
      fill=blue!14,
      line width=1.0pt,
      minimum width={\n1+10pt},
      minimum height=10pt,
      rounded corners=5pt,
      inner sep=0pt,
      rotate=\n2
    ] at ($ (#1)!0.5!(#2) $) {};
  \path
    let
      \p1=(#1),
      \p2=(#2),
      \n1={veclen(\x2-\x1,\y2-\y1)},
      \n2={atan2(\y2-\y1,\x2-\x1)}
    in
    node[
      draw=none,
      pattern=north east lines,
      pattern color=blue!45!cyan,
      minimum width={\n1+10pt},
      minimum height=10pt,
      rounded corners=5pt,
      inner sep=0pt,
      rotate=\n2
    ] at ($ (#1)!0.5!(#2) $) {};
}

\newcommand{\pdaxeslabels}{%
  \draw[line width=0.9pt,->] (0,0) -- (1.04,0)
    node[below right] {birth};
  \draw[line width=0.9pt,->] (0,0) -- (0,1.04)
    node[above left,rotate=90] {death};
  \draw[line width=0.9pt] (0,0) -- (1,1);
}

\newcommand{\pdaxes}{%
  \draw[line width=0.9pt,->] (0,0) -- (1.04,0);
  \draw[line width=0.9pt,->] (0,0) -- (0,1.04);
  \draw[line width=0.9pt] (0,0) -- (1,1);
}

\newcommand{\pointguides}[2]{%
  \draw[guide] (#1,#2) -- (#1,#1);
  \draw[guide] (#1,#2) -- (#2,#2);
}

\begin{scope}

\pdaxeslabels

\coordinate (bottomL) at (0.50,0.60);
\coordinate (topL) at (0.20,0.90);
\coordinate (middleRightL) at (0.45,0.825);
\coordinate (middleLeftL) at (0.375,0.85);

\pointguides{0.50}{0.60}
\pointguides{0.20}{0.90}
\pointguides{0.45}{0.825}
\pointguides{0.375}{0.85}

\ordertwocapsule{bottomL}{topL}
\ordertwocapsule{middleRightL}{middleLeftL}

\path[orderthree]
  (0.155,0.925)
  .. controls (0.175,0.975) and (0.265,0.980) .. (0.345,0.930)
  .. controls (0.405,0.895) and (0.475,0.895) .. (0.505,0.855)
  .. controls (0.545,0.800) and (0.585,0.675) .. (0.555,0.590)
  .. controls (0.535,0.535) and (0.485,0.520) .. (0.445,0.565)
  .. controls (0.390,0.625) and (0.345,0.720) .. (0.285,0.765)
  .. controls (0.230,0.805) and (0.165,0.825) .. (0.145,0.875)
  .. controls (0.135,0.895) and (0.140,0.915) .. (0.155,0.925)
  -- cycle;

\node[orderone] at (bottomL) {};
\node[orderone] at (topL) {};
\node[orderone] at (middleRightL) {};
\node[orderone] at (middleLeftL) {};

\node at (0.50,-0.13) {(a)};
\end{scope}

\begin{scope}[xshift=4.30cm]

\pdaxes

\coordinate (bottomR) at (0.50,0.60);
\coordinate (topR) at (0.20,0.90);
\coordinate (unusedR) at (0.45,0.825);
\coordinate (selectedMiddleR) at (0.375,0.85);

\pointguides{0.50}{0.60}
\pointguides{0.20}{0.90}
\pointguides{0.45}{0.825}
\pointguides{0.375}{0.85}

\ordertwocapsule{bottomR}{topR}
\ordertwocapsule{bottomR}{selectedMiddleR}

\path[orderthree]
  (0.155,0.925)
  .. controls (0.175,0.975) and (0.235,0.970) .. (0.285,0.925)
  .. controls (0.325,0.890) and (0.355,0.885) .. (0.390,0.880)
  .. controls (0.415,0.875) and (0.420,0.850) .. (0.405,0.815)
  .. controls (0.435,0.750) and (0.515,0.660) .. (0.545,0.605)
  .. controls (0.570,0.555) and (0.525,0.525) .. (0.480,0.560)
  .. controls (0.420,0.610) and (0.360,0.715) .. (0.305,0.785)
  .. controls (0.255,0.845) and (0.195,0.855) .. (0.155,0.895)
  .. controls (0.145,0.905) and (0.145,0.915) .. (0.155,0.925)
  -- cycle;

\node[orderone] at (bottomR) {};
\node[orderone] at (topR) {};
\node[orderone] at (unusedR) {};
\node[orderone] at (selectedMiddleR) {};

\node at (0.50,-0.13) {(b)};
\end{scope}

\end{tikzpicture}
}

\caption{
Illustration of two order-\(3\) persistence intervals constructed from the same four order-\(1\) intervals. Black points denote order-\(1\) intervals, patterned light-blue regions denote order-\(2\) intervals, and medium-blue regions denote order-\(3\) intervals. The order-\(3\) interval in (a) uses all four order-\(1\) intervals, while the order-\(3\) interval in (b) uses only three.
}
\label{fig:higher-order-intervals}
\end{wrapfigure}

Definition~\ref{def:higher-order-persistence-intervals} and
Lemma~\ref{lem:higher-order-metric} therefore define a metric on
\(X^{(n)}/A^{(n)}\), from which we obtain the Wasserstein metric on
\(D^{(n)}(X,A)\). For \(\Gamma,\Lambda\in D^{(n)}(X,A)\), define
\[
W_p^{(n)}(\Gamma,\Lambda)
:=
\inf_{\sigma}
\left\|
\bigl(d^{(n)}(a_i,b_i)\bigr)_i
\right\|_p,
\]
where \(\sigma=\sum_i(a_i,b_i)\) ranges over the matchings of
\(\Gamma\) and \(\Lambda\).

\paragraph{Virtual persistence diagrams.}

The Wasserstein distance summarizes the difference between two higher-order persistence diagrams by a single scalar. A more detailed comparison would retain the multiplicity difference at each order-\(n\) persistence interval. We obtain this comparison by subtracting the multiplicity of each interval in one diagram from its multiplicity in the other.

Virtual persistence diagrams represent this interval-wise comparison through signed multiplicities \cite{bubenik2022virtual}. For \(\Gamma,\Lambda\in D^{(n)}(X,A)\), the coefficient of an order-\(n\) persistence interval \(I\) in \(\Gamma-\Lambda\) equals the difference between the multiplicity of \(I\) in \(\Gamma\) and its multiplicity in \(\Lambda\). Thus, \(\Gamma-\Lambda\) records the difference between \(\Gamma\) and \(\Lambda\) at each order-\(n\) persistence interval. For each \(n\geq 1\), let \(K^{(n)}(X,A)\) denote the corresponding order-\(n\) virtual persistence diagram group.

For \(p=1\), the metric \(W_1^{(n)}\) extends to the translation-invariant metric~\cite{bubenik2022virtual}
\[
\rho^{(n)}
\bigl(
\Gamma-\Lambda,
\Gamma'-\Lambda'
\bigr)
:=
W_1^{(n)}
\bigl(
\Gamma+\Lambda',
\Gamma'+\Lambda
\bigr).
\]

Our harmonic aggregation method uses the Fourier transform of the virtual persistence diagram group \cite{fanningaktas2025rkhsvpd,fanningaktas2026randomwalks}. The Fourier analysis used here requires the virtual persistence diagram group to be locally compact and abelian. Corollary~3.3 of \cite{fanningaktas2026banachrkhs} shows that the virtual persistence diagram group \(K(X,A)\) is locally compact and abelian when \((X/A,d_1,[A])\) is uniformly discrete. Each higher-order persistence diagram space is also a persistence diagram space, so this result applies at every order. We next show that the required uniform discreteness propagates from one order to the next.

\begin{proposition}\label{prop:uniform-discreteness-levels}
For \(n\geq 1\), if \(\bigl(X^{(n)}/A^{(n)},d^{(n)}\bigr)\) is uniformly discrete, then
\(\bigl(X^{(n+1)}/A^{(n+1)},d^{(n+1)}\bigr)\) is uniformly discrete.
\end{proposition}

The proof is deferred to Appendix~\ref{app:proof-uniform-discreteness-levels}.

\begin{remark}\label{rem:higher-order-lca}
Applying Corollary~3.3 of \cite{fanningaktas2026banachrkhs} to the order-\(n\) metric pair, the virtual persistence diagram group \(K^{(n)}(X,A)\) is locally compact and discrete whenever \(\bigl(X^{(n)}/A^{(n)},d^{(n)}\bigr)\) is uniformly discrete. This gives the locally compact abelian group required by the Fourier and Fourier--Stieltjes transforms.
\end{remark}

\paragraph{Aggregation.}\label{par:aggregation}

Virtual persistence diagrams assign a signed coefficient to each order-\(n\) interval, but these coefficients alone do not record which intervals occur together within the same observation. To preserve relationships between intervals that occur together, we represent two intervals by the ordered pair \((u,v)\), where \(u\preceq^{(n)}v\) and \(u\not\sim^{(n)}v\). Since every observation uses the same order-\(n\) interval space, the same ordered pair has the same interpretation throughout the collection.

For an ordinary persistence diagram, multiplicities \(a\) and \(b\) of \(u\) and \(v\) produce \(ab\) ordered pairs. We extend this counting rule to a virtual persistence diagram \(\Gamma\) by defining the signed coefficient of \((u,v)\) as \(\Gamma(u)\Gamma(v)\). Since \((u,v)\) is an order-\((n+1)\) interval, our construction represents each observation by an order-\((n+1)\) virtual persistence diagram.

We first construct one order-\((n+1)\) virtual persistence diagram from each observation and then aggregate these diagrams across the observations. For each observation, the ordered pair represents the relationship, and the corresponding product gives its signed contribution. The resulting sum gives the total signed contribution of that relationship across all observations.

We also define a mean aggregate to remove scaling due to the number of observations. Because averaging generally produces real-valued coefficients, the mean aggregate lies in the space of persistence diagrams with real-valued multiplicities. We denote this space by \(V^{(n+1)}(X,A)\).

\begin{definition}\label{def:aggregation}
Fix integers \(n,m\geq 1\) and \(\Gamma_1,\ldots,\Gamma_m\in K^{(n)}(X,A)\).

\begin{enumerate}
\item The sum aggregate is
\[
\mathcal A_n(\Gamma_1,\ldots,\Gamma_m)
:=
\sum_{r=1}^m
\sum_{\substack{
u,v\in\operatorname{supp}(\Gamma_r)\\
u\preceq^{(n)}v,\;
u\not\sim^{(n)}v
}}
\Gamma_r(u)\Gamma_r(v)\,[(u,v)]
\in
K^{(n+1)}(X,A).
\]

\item The mean aggregate is
\[
\overline{\mathcal A}_n(\Gamma_1,\ldots,\Gamma_m)
:=
\frac{1}{m}
\mathcal A_n(\Gamma_1,\ldots,\Gamma_m)
\in
V^{(n+1)}(X,A).
\]
\end{enumerate}
\end{definition}

\begin{figure}[t]
\centering
\includegraphics[width=\linewidth]{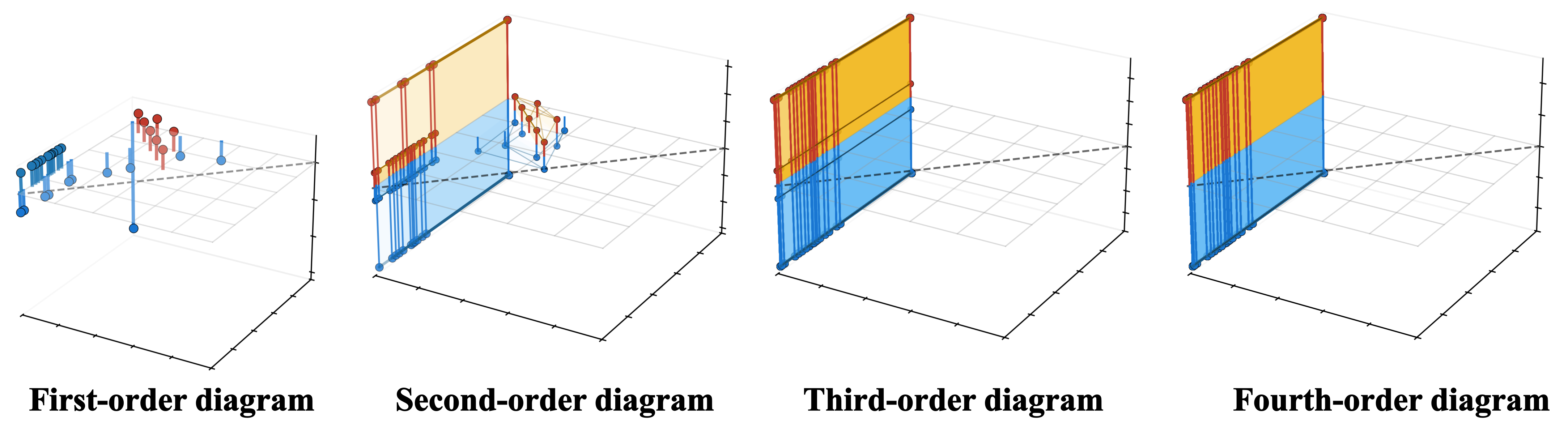}
\caption{
Blue pins represent order-\(1\) persistence intervals in \(H_0\). Red pins represent order-\(1\) persistence intervals in \(H_1\). The height of each pin represents its signed multiplicity. Yellow surfaces represent higher-order persistence intervals.
}
\label{fig:higher-order-single}
\end{figure}

\paragraph{Harmonic aggregation.}

For the remainder of this section, assume that \(\bigl(X^{(1)}/A^{(1)},d^{(1)}\bigr)\) is uniformly discrete with minimum distance \(\varepsilon>0\). Proposition~\ref{prop:uniform-discreteness-levels} implies that \(\bigl(X^{(n)}/A^{(n)},d^{(n)}\bigr)\) has minimum distance at least \(\varepsilon\) for every \(n\geq1\), and Remark~\ref{rem:higher-order-lca} then implies that \(K^{(n)}(X,A)\) is locally compact and discrete for every \(n\geq1\).

We evaluate the aggregate in Definition~\ref{def:aggregation} at characters in the Pontryagin dual \(\widehat{K^{(n+1)}(X,A)}\). Let \(\psi:X^{(n)}\to\mathbb R\) be constant on each equivalence class of \(\sim^{(n)}\). Then \(\psi(v)-\psi(u)=\psi(v')-\psi(u')\) whenever \((u,v)\sim^{(n+1)}(u',v')\). For every order-\((n+1)\) interval \((u,v)\), set
\begin{equation}\label{eq:phase-character}
\chi_{\psi}([(u,v)])
:=
\exp\left(i\bigl(\psi(v)-\psi(u)\bigr)\right).
\end{equation}
The compatibility of these differences makes \(\chi_\psi([(u,v)])\) independent of the representative of \([(u,v)]\). If \((u,v)\in A^{(n+1)}\), then \(u\sim^{(n)}v\), so \(\psi(u)=\psi(v)\) and Equation~\ref{eq:phase-character} assigns \(1\) to the basepoint. Since \(K^{(n+1)}(X,A)\) is free abelian on the non-basepoint order-\((n+1)\) interval classes, this assignment extends uniquely to a homomorphism \(\chi_\psi:K^{(n+1)}(X,A)\to\mathbb T\). By Remark~\ref{rem:higher-order-lca}, \(K^{(n+1)}(X,A)\) is discrete, so \(\chi_\psi\) is continuous and belongs to \(\widehat{K^{(n+1)}(X,A)}\).

Definition~\ref{def:aggregation} gives
\begin{equation}\label{eq:aggregate-character}
\chi_{\psi}\!\left(
\mathcal A_n(\Gamma_1,\ldots,\Gamma_m)
\right)
=
\exp\left(
i
\sum_{j=1}^{m}
\sum_{\substack{
u,v\in\operatorname{supp}(\Gamma_j)\\
u\preceq^{(n)}v,\;
u\not\sim^{(n)}v
}}
\Gamma_j(u)\Gamma_j(v)
\bigl(\psi(v)-\psi(u)\bigr)
\right).
\end{equation}
We evaluate the exponent in Equation~\ref{eq:aggregate-character} by computing weighted sums over intervals that satisfy the preorder, without constructing the aggregate.

\begin{lemma}\label{lem:order-dimension-levels}
If \(\preceq^{(n)}\) has finite order dimension \(r_n\), then \(\preceq^{(n+1)}\) has order dimension at most \(2r_n\).
\end{lemma}

For the following complexity bounds, we use the unit-cost RAM model and assume constant-time arithmetic and constant-time access to stored coefficients, values of the representing maps, and values of the functions that define the characters.

\begin{theorem}\label{thm:implicit-harmonic-complexity}
Let \(r_n\) be the fixed order dimension of \(\preceq^{(n)}\), and suppose that we have the values of its \(r_n\) representing maps on every input support. Given the values of \(\psi\) on the input supports, we can compute Equation~\ref{eq:aggregate-character} in
\begin{equation}\label{eq:implicit-character-complexity}
O\left(
\sum_{j=1}^{m}
|\operatorname{supp}(\Gamma_j)|
\log^{\max\{1,r_n-1\}}
|\operatorname{supp}(\Gamma_j)|
\right)
\end{equation}
time without constructing \(\mathcal A_n(\Gamma_1,\ldots,\Gamma_m)\). Given the values on the input supports of \(K\) functions that define characters by Equation~\ref{eq:phase-character}, we can compute the corresponding \(K\) character evaluations in
\begin{equation}\label{eq:multiple-character-complexity}
O\left(
\sum_{j=1}^{m}
|\operatorname{supp}(\Gamma_j)|
\log^{\max\{1,r_n-1\}}
|\operatorname{supp}(\Gamma_j)|
+
K\sum_{j=1}^{m}
|\operatorname{supp}(\Gamma_j)|
\right)
\end{equation}
time.
\end{theorem}

Direct explicit construction tests every ordered pair within each input support and therefore requires \(O(\sum_{j=1}^{m}|\operatorname{supp}(\Gamma_j)|^2)\) time. We direct an edge from each interval in a recursive representation to each of its at most two immediate lower-order parents. Suppose that the underlying undirected graph has \(c\) connected components, each with a root, and let \(N\) be the maximum number of edges in a directed path from a root. Every directed edge decreases the order, and different branches can share lower-order intervals. Unfolding each component into a rooted tree with branching degree at most two can therefore only increase the number of interval occurrences. Each unfolding contains at most \(2^{N+1}-1\) interval occurrences. Hence, for an order-\(n\) input \(\xi\),
\[
|\operatorname{supp}(\xi)|
\leq
c(2^{N+1}-1)
=
O(c2^N).
\]
For the corresponding aggregate of \(\xi\), this support bound gives
\begin{equation}\label{eq:explicit-implicit-complexity}
\underbrace{O(c^2 4^N)}_{\text{direct explicit construction}}
\qquad\text{and}\qquad
\underbrace{
O\left(
c2^N
\left[
\log^{\max\{1,r_n-1\}}(c2^N)+K
\right]
\right)
}_{K\text{ character evaluations}}.
\end{equation}
If \(r_n\leq r\) for a fixed \(r\geq2\), one character evaluation of this aggregate requires \(O(c2^N\log^{r-1}(c2^N))\) time.

\begin{definition}\label{def:harmonic-aggregation}
For \(n\geq1\), \(\Gamma\in V^{(n)}(X,A)\), and \(\chi_L,\chi_R\in\widehat{K^{(n)}(X,A)}\), define the harmonic aggregation coordinate of \(\Gamma\) at \((\chi_L,\chi_R)\) by
\begin{equation}\label{eq:harmonic-aggregation}
\operatorname{HA}_n(\Gamma;\chi_L,\chi_R)
:=
\frac{1}{2}
\left(
\sum_{u\in\operatorname{supp}(\Gamma)}
\Gamma(u)\chi_L([u])
+
\sum_{u\in\operatorname{supp}(\Gamma)}
\Gamma(u)\chi_R([u])
\right).
\end{equation}
\end{definition}

\begin{theorem}\label{thm:harmonic-aggregation-factorization}
Let \(\Gamma_1,\ldots,\Gamma_m\in K^{(n)}(X,A)\), and let \(\chi_{\psi_L}\) and \(\chi_{\psi_R}\) be characters defined by Equation~\ref{eq:phase-character}. Then
\begin{equation}\label{eq:harmonic-aggregation-factorization}
\begin{aligned}
\operatorname{HA}_{n+1}(
\mathcal A_n(\Gamma_1,\ldots,&\Gamma_m);
\chi_{\psi_L},\chi_{\psi_R})\\
=
\frac{1}{2}
\sum_{j=1}^{m}
\Bigg[
&\sum_{v\in\operatorname{supp}(\Gamma_j)}
\Gamma_j(v)e^{i\psi_L(v)}
\sum_{\substack{
u\in\operatorname{supp}(\Gamma_j)\\
u\preceq^{(n)}v,\;
u\not\sim^{(n)}v
}}
\Gamma_j(u)e^{-i\psi_L(u)}
\\
+&\sum_{v\in\operatorname{supp}(\Gamma_j)}
\Gamma_j(v)e^{i\psi_R(v)}
\sum_{\substack{
u\in\operatorname{supp}(\Gamma_j)\\
u\preceq^{(n)}v,\;
u\not\sim^{(n)}v
}}
\Gamma_j(u)e^{-i\psi_R(u)}
\Bigg].
\end{aligned}
\end{equation}
\end{theorem}

Computing the left-hand side of Equation~\ref{eq:harmonic-aggregation-factorization} by explicitly constructing the aggregate requires \(O(c^2 4^N)\) time. For fixed \(n\) and \(r_n\leq r\) for a fixed \(r\geq2\), the weighted sums on the right-hand side require \(O(c2^N\log^{r-1}(c2^N))\) time for one harmonic aggregation coordinate and \(O(Kc2^N\log^{r-1}(c2^N))\) time for \(K\) independently parameterized coordinates without constructing the aggregate.

\begin{theorem}\label{thm:harmonic-aggregation-stability}
For every \(n\geq1\) and \(\chi_L,\chi_R\in\widehat{K^{(n)}(X,A)}\), \(\operatorname{HA}_n(\,\cdot\,;\chi_L,\chi_R)\) is Lipschitz stable with constant \(2/\varepsilon\).
\end{theorem}

Appendix~\ref{app:deferred-proofs} contains the proofs of Lemma~\ref{lem:order-dimension-levels}, Theorem~\ref{thm:implicit-harmonic-complexity}, Theorem~\ref{thm:harmonic-aggregation-factorization}, and Theorem~\ref{thm:harmonic-aggregation-stability}.

\section{Experiments}\label{sec:experiments}

\begin{wraptable}{r}{0.55\textwidth}
    \centering
    \footnotesize
    \setlength{\tabcolsep}{4pt}
    \renewcommand{\arraystretch}{1.05}
    \caption{Pooled out-of-fold classification performance on the random network model benchmark. Bold values are the largest, and underlined values are the second largest.}
    \label{tab:benchmark-classification}
    \begin{tabular}{lcc}
        \toprule
        Method & Accuracy & Macro-\(F_1\) \\
        \midrule
        HA & \(\mathbf{0.954} \pm 0.013\) & \(\mathbf{0.954} \pm 0.013\) \\
        PI & \(\underline{0.940} \pm 0.014\) & \(\underline{0.939} \pm 0.014\) \\
        PL & \(0.786 \pm 0.025\) & \(0.785 \pm 0.025\) \\
        PersLay & \(0.889 \pm 0.020\) & \(0.889 \pm 0.019\) \\
        \bottomrule
    \end{tabular}
\end{wraptable}

In this section, we evaluate the predictive performance of harmonic aggregation (HA) from Definition~\ref{def:harmonic-aggregation}. Throughout the experiments, we use the preorder induced by the rank function. The \(256\)-coordinate representation uses \(63\) order-\(1\) HA coordinates \(\operatorname{HA}_1(\Gamma;\chi,\chi)\), \(63\) order-\(1\) HA coordinates that use each coefficient \(\Gamma(u)\) multiplied by the persistence of \(u\), \(63\) order-\(2\) HA coordinates \(\operatorname{HA}_2(\mathcal A_1(\Gamma);\chi_\psi,\chi_\psi)\), and \(63\) order-\(2\) HA coordinates that use each coefficient \(\Gamma(u)\) multiplied by the persistence of \(u\). The remaining four coordinates are the exact evaluations of these four constructions at the trivial character. We compare this representation with three established persistence-diagram vectorizations on the random network model benchmark and real-world graph-classification benchmarks. Appendix~\ref{app:implementation-details} gives the complete implementation details.

\subsection{Random network model benchmark}\label{subsec:benchmark}

We construct each batch from four independent graphs generated by the same random network model. Each of the ten models contributes \(100\) batches, for a total of \(1000\) batches and \(4000\) graphs. Every graph has \(100\) vertices and \(200\) edges. We compute an \(H_1\) persistence diagram for each graph and use these same persistence diagrams to compute persistence landscapes (PL), persistence images (PI), PersLay, and HA. Appendix~\ref{app:implementation-details} specifies the random network models, filtration procedure, and representation parameters.

For each batch with persistence diagrams \(\Gamma_1,\ldots,\Gamma_4\), we average their \(256\)-coordinate HA representations. We evaluate the order-\(2\) HA coordinates using the factorization in Theorem~\ref{thm:harmonic-aggregation-factorization} without explicitly constructing the aggregate.

We evaluate all representations with the same nested cross-validation protocol. The protocol uses \(10\) outer folds, \(5\) inner folds, and a support vector machine with an RBF kernel. We fit every transformation and preprocessing operation using only the training data for the corresponding fold. We report pooled out-of-fold accuracy and macro-\(F_1\) with \(95\%\) bootstrap confidence intervals (CIs), and we use paired bootstrap CIs for paired comparisons.

Table~\ref{tab:benchmark-classification} shows that HA has a higher classification accuracy on the random network model benchmark than PI, PL, and PersLay.

\paragraph{Ablation study.}

We consider three ablations that use intervals from different observations for the order-\(2\) coordinates, use all strictly inequivalent ordered pairs for the order-\(2\) coordinates, or use only the order-\(1\) construction and its exact evaluations at the trivial character.

The cross-observation ablation pairs each persistence diagram \(\Gamma\) in a batch with a persistence diagram \(\Lambda\) from the same batch using a fixed-point-free permutation. It uses the same order-\(1\) coordinates and corresponding exact evaluations as HA. For the order-\(2\) coordinates, it uses \(u\in\operatorname{supp}(\Gamma)\) and \(v\in\operatorname{supp}(\Lambda)\) subject to \(u\preceq^{(1)}v\) and \(u\not\sim^{(1)}v\), and computes
\begin{equation}\label{eq:cross-observation-quadratic}
\sum_{\substack{
u\in\operatorname{supp}(\Gamma),\,
v\in\operatorname{supp}(\Lambda)\\
u\preceq^{(1)}v,\;
u\not\sim^{(1)}v
}}
\Gamma(u)\Lambda(v)
\overline{\chi([u])}\chi([v]).
\end{equation}
The other \(63\) order-\(2\) coordinates use each coefficient \(\Gamma(u)\) and \(\Lambda(v)\) multiplied by the persistence of its interval. The two order-\(2\) exact coordinates use the corresponding sums at the trivial character. We independently sample ten fixed-point-free permutations. The resulting accuracies range from \(0.942\) to \(0.949\), compared with the HA accuracy of \(0.954\).

The no-preorder ablation uses the same order-\(1\) coordinates and corresponding exact evaluations as HA. For the order-\(2\) coordinates, it uses all ordered pairs \(u,v\in\operatorname{supp}(\Gamma)\) such that \(u\not\sim^{(1)}v\), and computes
\begin{equation}\label{eq:no-preorder-quadratic}
\sum_{\substack{
u,v\in\operatorname{supp}(\Gamma)\\
u\not\sim^{(1)}v
}}
\Gamma(u)\Gamma(v)
\overline{\chi([u])}\chi([v]).
\end{equation}
The other \(63\) order-\(2\) coordinates use each coefficient \(\Gamma(u)\) multiplied by the persistence of \(u\). The two order-\(2\) exact coordinates use the corresponding sums at the trivial character. Its accuracy is \(0.947\), with a paired difference from HA of \(-0.007\) and a paired \(95\%\) CI of \([-0.017,0.003]\).

The linear ablation uses \(127\) order-\(1\) HA coordinates \(\operatorname{HA}_1(\Gamma;\chi,\chi)\), \(127\) order-\(1\) HA coordinates that use each coefficient \(\Gamma(u)\) multiplied by the persistence of \(u\), and the corresponding two exact evaluations at the trivial character. Its accuracy is \(0.946\), with a paired difference from HA of \(-0.008\) and a paired \(95\%\) CI of \([-0.019,0.003]\).

\paragraph{Sensitivity.}

We vary the batch size over \(b\in\{1,\ldots,7\}\), with \(b=4\) as the reference. Accuracy increases from \(0.820\) at \(b=1\) to \(0.954\) at \(b=4\) and \(0.970\) at \(b=7\). The paired \(95\%\) CIs for the accuracy differences lie strictly below zero for \(b=1,2,3\), contain zero for \(b=5\), and lie strictly above zero for \(b=6,7\). These results show that HA is sensitive to small batch sizes.

We vary the character count over \(K\in\{64,128,256,512,1024\}\) using ten replicates, with nested samples across the values of \(K\) within each replicate. The mean accuracy over the ten replicates ranges from \(0.9478\) to \(0.9501\), a span of \(0.0023\), and every paired \(95\%\) CI for the difference from benchmark HA contains zero. These results show that HA is insensitive to the character count over the tested range.

\paragraph{Robustness.}

We evaluate edge deletion, edge insertion, degree-preserving rewiring, and filtration noise with ten realizations at each severity. Across the tested severity ranges, the maximum mean accuracy loss is \(0.0495\) for filtration noise, compared with \(0.0608\) for edge insertion, \(0.0886\) for degree-preserving rewiring, and \(0.1012\) for edge deletion.

\subsection{Real-world data benchmark}\label{subsec:real-world-data}

We evaluate HA, PI, PL, and PersLay on the real-world graph-classification datasets AIDS, BZR, COX2, MUTAG, Mutagenicity, NCI1, NCI109, and PROTEINS \cite{Morris+2020}. All four representations use the same \(H_1\) persistence diagrams. We evaluate bag sizes \(b\in\{1,\ldots,7\}\). For all four methods, we represent each bag by the coordinatewise mean of its graph-level representations. At each fixed \(b\), we use the same bags for all methods and perform downstream model selection separately for each method. We train the PersLay representation within each training partition. HA uses the same sampled characters across all datasets and bag sizes. Appendix~\ref{app:implementation-details} specifies the construction of the persistence diagrams, bag construction, PersLay training, nested cross-validation, and preprocessing.

\begin{table}[t!b]
\centering
\footnotesize
\setlength{\tabcolsep}{1.8pt}
\renewcommand{\arraystretch}{1.03}
\caption{Pooled out-of-fold accuracy at bag sizes \(b=5\), \(b=6\), and \(b=7\). Bold values are the largest within each dataset and bag size, and underlined values are the largest values strictly below the maximum. Average ranks use rank \(1\) for the largest accuracy within each dataset and average ranks for ties.}
\label{tab:real-world-endpoints}
\begin{tabular}{lcccccccccccc}
\toprule
& \multicolumn{4}{c}{\(b=5\)}
& \multicolumn{4}{c}{\(b=6\)}
& \multicolumn{4}{c}{\(b=7\)} \\
\cmidrule(lr){2-5}\cmidrule(lr){6-9}\cmidrule(lr){10-13}
Dataset
& HA & PI & PL & PersLay
& HA & PI & PL & PersLay
& HA & PI & PL & PersLay \\
\midrule
AIDS
& \(\mathbf{0.988}\) & \(0.980\) & \(\underline{0.985}\) & \(0.978\)
& \(\mathbf{0.991}\) & \(\mathbf{0.991}\) & \(0.988\) & \(\mathbf{0.991}\)
& \(\mathbf{0.996}\) & \(\mathbf{0.996}\) & \(\underline{0.993}\) & \(\mathbf{0.996}\) \\
BZR
& \(\underline{0.840}\) & \(\underline{0.840}\) & \(\mathbf{0.907}\) & \(0.827\)
& \(\underline{0.836}\) & \(0.803\) & \(\mathbf{0.918}\) & \(\underline{0.836}\)
& \(\underline{0.873}\) & \(0.855\) & \(\mathbf{0.945}\) & \(0.836\) \\
COX2
& \(\mathbf{0.778}\) & \(\mathbf{0.778}\) & \(\mathbf{0.778}\) & \(\mathbf{0.778}\)
& \(\mathbf{0.800}\) & \(\mathbf{0.800}\) & \(\mathbf{0.800}\) & \(\underline{0.787}\)
& \(\mathbf{0.823}\) & \(\mathbf{0.823}\) & \(\underline{0.806}\) & \(\mathbf{0.823}\) \\
MUTAG
& \(\mathbf{1.000}\) & \(\mathbf{1.000}\) & \(\underline{0.971}\) & \(\mathbf{1.000}\)
& \(\mathbf{1.000}\) & \(\mathbf{1.000}\) & \(\underline{0.933}\) & \(\mathbf{1.000}\)
& \(\mathbf{1.000}\) & \(\mathbf{1.000}\) & \(\mathbf{1.000}\) & \(\mathbf{1.000}\) \\
Mutagenicity
& \(0.726\) & \(\underline{0.734}\) & \(\mathbf{0.742}\) & \(0.726\)
& \(\underline{0.749}\) & \(0.744\) & \(0.743\) & \(\mathbf{0.756}\)
& \(\underline{0.746}\) & \(0.737\) & \(\mathbf{0.761}\) & \(0.733\) \\
NCI1
& \(0.726\) & \(\mathbf{0.733}\) & \(0.696\) & \(\underline{0.728}\)
& \(\mathbf{0.741}\) & \(\underline{0.740}\) & \(0.656\) & \(0.728\)
& \(0.771\) & \(\mathbf{0.783}\) & \(0.617\) & \(\underline{0.776}\) \\
NCI109
& \(\mathbf{0.753}\) & \(\mathbf{0.753}\) & \(0.704\) & \(\underline{0.747}\)
& \(\underline{0.771}\) & \(\mathbf{0.772}\) & \(0.699\) & \(\mathbf{0.772}\)
& \(\mathbf{0.795}\) & \(\underline{0.790}\) & \(0.737\) & \(0.788\) \\
PROTEINS
& \(\mathbf{0.805}\) & \(0.736\) & \(0.727\) & \(\underline{0.755}\)
& \(\mathbf{0.822}\) & \(0.735\) & \(0.741\) & \(\underline{0.768}\)
& \(\mathbf{0.850}\) & \(\underline{0.791}\) & \(0.758\) & \(0.784\) \\
\midrule
Average rank
& \(\mathbf{2.125}\) & \(2.188\) & \(2.813\) & \(2.875\)
& \(\mathbf{1.938}\) & \(2.563\) & \(3.250\) & \(2.250\)
& \(\mathbf{1.938}\) & \(2.188\) & \(3.063\) & \(2.813\) \\
\bottomrule
\end{tabular}
\end{table}

Table~\ref{tab:real-world-endpoints} gives the accuracy at \(b=5\), \(b=6\), and \(b=7\). At \(b=1\), the task is graph classification rather than MIL. At small values of \(b\), each bag contains only a few instances, so aggregation has only a limited effect on the classification task. For large values of \(b\), the accuracy on several of the smaller datasets saturates. The average ranks over the eight datasets and the three reported bag sizes are \(2.000\) for HA, \(2.313\) for PI, \(2.646\) for PersLay, and \(3.042\) for PL.

\FloatBarrier
\section{Conclusion}\label{sec13}

We introduced higher-order persistence diagrams to represent relationships between lower-order intervals and an aggregation that retains information about which intervals occur together within the same observation after aggregating across observations. Harmonic aggregation uses Fourier analysis to represent the higher-order aggregate without explicitly constructing it. For fixed order dimension \(r\), we can compute one harmonic aggregation coordinate in \(O(c2^N\log^{r-1}(c2^N))\) time, compared with \(O(c^2 4^N)\) time for explicit construction of the aggregate. Assuming uniform discreteness, the virtual persistence diagram groups are locally compact and abelian, so we can define and use the Fourier transform on them. Harmonic aggregation has the largest accuracy on the random network model benchmark and the best average rank across the real-world benchmark, is insensitive to the number of HA coordinates, and is most robust to filtration noise among the tested perturbations. Future work will seek to extend harmonic aggregation to amenable virtual persistence diagram groups.

\section*{Declarations}

\begin{itemize}
\item \textbf{Competing Interests} The authors declare that they have no competing interests.
\item \textbf{Funding} This research received no external funding.
\item \textbf{Data Availability} Not applicable.
\item \textbf{Code Availability} The implementation used in this work is available at \url{https://github.com/cfanning8/higher-order-PDs}.
\item \textbf{Authors' Contributions}
C.F. developed the theoretical framework, conducted the experiments, and wrote the manuscript. M.E.A. supervised the project and provided feedback on the framework, experiments, and manuscript.
\item \textbf{AI Usage} The authors used an AI language model to rephrase dense analytic descriptions and to proofread and correct the language and grammar during manuscript preparation. The authors assume responsibility for all content in this publication.
\end{itemize}

\bibliography{sn-bibliography}

@article{adams2017persistence,
  author  = {Henry Adams and Tegan Emerson and Michael Kirby and Rachel Neville and Chris Peterson and Patrick Shipman and Sofya Chepushtanova and Eric Hanson and Francis Motta and Lori Ziegelmeier},
  title   = {{Persistence Images: A Stable Vector Representation of Persistent Homology}},
  journal = {Journal of Machine Learning Research},
  year    = {2017},
  volume  = {18},
  number  = {8},
  pages   = {1--35},
  issn    = {1532-4435},
  url     = {https://jmlr.org/papers/v18/16-337.html},
}

@article{ali2023survey,
  author  = {Dashti Ali and Aras Asaad and Maria-Jose Jimenez and Vidit Nanda and Eduardo Paluzo-Hidalgo and Manuel Soriano-Trigueros},
  title   = {{A Survey of Vectorization Methods in Topological Data Analysis}},
  journal = {IEEE Transactions on Pattern Analysis and Machine Intelligence},
  year    = {2023},
  volume  = {45},
  number  = {12},
  pages   = {14069--14080},
  doi     = {10.1109/TPAMI.2023.3308391},
  url     = {https://ieeexplore.ieee.org/document/10148529},
}

@inproceedings{carriere2020perslay,
  title     = {{PersLay}: A Neural Network Layer for Persistence Diagrams and New Graph Topological Signatures},
  author    = {Carri{\`e}re, Mathieu and Chazal, Fr{\'e}d{\'e}ric and Ike, Yuichi and Lacombe, Th{\'e}o and Royer, Martin and Umeda, Yuhei},
  booktitle = {Proceedings of the Twenty Third International Conference on Artificial Intelligence and Statistics},
  editor    = {Chiappa, Silvia and Calandra, Roberto},
  series    = {Proceedings of Machine Learning Research},
  volume    = {108},
  pages     = {2786--2796},
  year      = {2020},
  publisher = {PMLR},
  address   = {Online},
  url       = {https://proceedings.mlr.press/v108/carriere20a.html}
}

@article{10.1145/358841.358850,
  author  = {Jon Louis Bentley},
  title   = {{Multidimensional Divide-and-Conquer}},
  journal = {Communications of the ACM},
  year    = {1980},
  volume  = {23},
  number  = {4},
  pages   = {214--229},
  month   = apr,
  issn    = {0001-0782},
  doi     = {10.1145/358841.358850},
  url     = {https://doi.org/10.1145/358841.358850},
}

@article{betthauser2022graded,
  author  = {Leo Betthauser and Peter Bubenik and Parker B. Edwards},
  title   = {{Graded Persistence Diagrams and Persistence Landscapes}},
  journal = {Discrete \& Computational Geometry},
  year    = {2022},
  volume  = {67},
  number  = {1},
  pages   = {203--230},
  doi     = {10.1007/s00454-021-00316-1},
  url     = {https://doi.org/10.1007/s00454-021-00316-1},
}

@inproceedings{borowa2021deep,
  author    = {Adriana Borowa and Dawid Rymarczyk and Dorota Ocho{\'n}ska and Monika Brzychczy-W{\l}och and Bartosz Zieli{\'n}ski},
  title     = {{Deep Learning Classification of Bacteria Clones Explained by Persistence Homology}},
  booktitle = {2021 International Joint Conference on Neural Networks},
  year      = {2021},
  pages     = {1--8},
  publisher = {IEEE},
  address   = {Shenzhen, China},
  doi       = {10.1109/IJCNN52387.2021.9533954},
  url       = {https://doi.org/10.1109/IJCNN52387.2021.9533954},
}

@article{bubenik2015statistical,
  author  = {Peter Bubenik},
  title   = {{Statistical Topological Data Analysis Using Persistence Landscapes}},
  journal = {Journal of Machine Learning Research},
  year    = {2015},
  volume  = {16},
  number  = {3},
  pages   = {77--102},
  issn    = {1532-4435},
  url     = {https://jmlr.org/papers/v16/bubenik15a.html},
}

@article{bubenik2022virtual,
  author  = {Peter Bubenik and Alex Elchesen},
  title   = {{Virtual Persistence Diagrams, Signed Measures, Wasserstein Distances, and Banach Spaces}},
  journal = {Journal of Applied and Computational Topology},
  year    = {2022},
  volume  = {6},
  number  = {4},
  pages   = {429--474},
  doi     = {10.1007/s41468-022-00091-9},
  url     = {https://doi.org/10.1007/s41468-022-00091-9},
}

@article{carbonneau2018multiple,
  author  = {Marc-Andr{\'e} Carbonneau and Veronika Cheplygina and Eric Granger and Ghyslain Gagnon},
  title   = {{Multiple Instance Learning: A Survey of Problem Characteristics and Applications}},
  journal = {Pattern Recognition},
  year    = {2018},
  volume  = {77},
  pages   = {329--353},
  doi     = {10.1016/j.patcog.2017.10.009},
  url     = {https://doi.org/10.1016/j.patcog.2017.10.009},
}

@article{DIETTERICH199731,
  author  = {Thomas G. Dietterich and Richard H. Lathrop and Tom{\'a}s Lozano-P{\'e}rez},
  title   = {{Solving the Multiple Instance Problem with Axis-Parallel Rectangles}},
  journal = {Artificial Intelligence},
  year    = {1997},
  volume  = {89},
  number  = {1--2},
  pages   = {31--71},
  issn    = {0004-3702},
  doi     = {10.1016/S0004-3702(96)00034-3},
  url     = {https://doi.org/10.1016/S0004-3702(96)00034-3},
}

@inproceedings{892133,
  author    = {Herbert Edelsbrunner and David Letscher and Afra Zomorodian},
  title     = {{Topological Persistence and Simplification}},
  booktitle = {Proceedings 41st Annual Symposium on Foundations of Computer Science},
  year      = {2000},
  pages     = {454--463},
  doi       = {10.1109/SFCS.2000.892133},
  url       = {https://ieeexplore.ieee.org/document/892133},
}

@article{fanningaktas2025rkhsvpd,
author  = {Charles Fanning and Mehmet Emin Aktas},
title   = {{Reproducing Kernel Hilbert Spaces for Virtual Persistence Diagrams}},
journal = {Journal of Applied and Computational Topology},
year    = {2026},
volume  = {10},
number  = {4},
pages   = {22},
doi     = {10.1007/s41468-026-00253-z},
url     = {https://doi.org/10.1007/s41468-026-00253-z},
}

@misc{fanningaktas2026banachrkhs,
  author        = {Charles Fanning and Mehmet Aktas},
  title         = {{Reproducing Kernel Hilbert Spaces on Banach Completions of Virtual Persistence Diagram Groups}},
  year          = {2026},
  eprint        = {2602.15153},
  archivePrefix = {arXiv},
  primaryClass  = {math.FA},
  doi           = {10.48550/arXiv.2602.15153},
  url           = {https://arxiv.org/abs/2602.15153},
}

@misc{fanningaktas2026randomwalks,
  author        = {Charles Fanning and Mehmet Aktas},
  title         = {{Random Walks on Virtual Persistence Diagrams}},
  year          = {2026},
  eprint        = {2603.02117},
  archivePrefix = {arXiv},
  primaryClass  = {math.PR},
  doi           = {10.48550/arXiv.2603.02117},
  url           = {https://arxiv.org/abs/2603.02117},
}

@article{hensel2021survey,
  author  = {Felix Hensel and Michael Moor and Bastian Rieck},
  title   = {{A Survey of Topological Machine Learning Methods}},
  journal = {Frontiers in Artificial Intelligence},
  year    = {2021},
  volume  = {4},
  pages   = {681108},
  doi     = {10.3389/frai.2021.681108},
  url     = {https://doi.org/10.3389/frai.2021.681108},
}

@inproceedings{ilse2018attention,
  author    = {Maximilian Ilse and Jakub Tomczak and Max Welling},
  title     = {{Attention-Based Deep Multiple Instance Learning}},
  booktitle = {Proceedings of the 35th International Conference on Machine Learning},
  series    = {Proceedings of Machine Learning Research},
  volume    = {80},
  pages     = {2127--2136},
  year      = {2018},
  publisher = {PMLR},
  address   = {Stockholm, Sweden},
  url       = {https://proceedings.mlr.press/v80/ilse18a.html},
}

@article{kazeminia2026topologicalinductivebiasfosters,
  author  = {Salome Kazeminia and Carsten Marr and Bastian Rieck},
  title   = {{Topological Inductive Bias Fosters Multiple Instance Learning in Data-Scarce Scenarios}},
  journal = {Transactions on Machine Learning Research},
  year    = {2026},
  url     = {https://openreview.net/forum?id=1hZy9ZjjCc},
}

@article{kimmemoli2021generalized,
  author  = {Woojin Kim and Facundo M{\'e}moli},
  title   = {{Generalized Persistence Diagrams for Persistence Modules over Posets}},
  journal = {Journal of Applied and Computational Topology},
  year    = {2021},
  volume  = {5},
  number  = {4},
  pages   = {533--581},
  doi     = {10.1007/s41468-021-00075-1},
  url     = {https://doi.org/10.1007/s41468-021-00075-1},
}

@inproceedings{Morris+2020,
  author        = {Christopher Morris and Nils M. Kriege and Franka Bause and Kristian Kersting and Petra Mutzel and Marion Neumann},
  title         = {{TUDataset: A Collection of Benchmark Datasets for Learning with Graphs}},
  booktitle     = {ICML 2020 Workshop on Graph Representation Learning and Beyond (GRL+ 2020)},
  year          = {2020},
  eprint        = {2007.08663},
  archivePrefix = {arXiv},
  primaryClass  = {cs.LG},
  doi           = {10.48550/arXiv.2007.08663},
  url           = {https://arxiv.org/abs/2007.08663},
}

@article{obeid2025prostate,
  author  = {Ahmad Obeid and Sajid Javed and Jorge Dias and Ibrahim Abe M. Elfadel and Naoufel Werghi},
  title   = {{Enhancing Prostate Cancer Subtyping: A Persistent Homology Approach in Multiple Instance Learning}},
  journal = {Engineered Science},
  year    = {2025},
  volume  = {38},
  pages   = {1905},
  doi     = {10.30919/es1905},
  url     = {https://doi.org/10.30919/es1905},
}

@inproceedings{obeid2025pmil,
  author    = {Ahmad Obeid and Anabia Sohail and Said Boumaraf and Xiabi Liu and Sajid Javed and Hasan Almarzouqi and Jorge Dias and Mohammed Bennamoun and Naoufel Werghi and Ibrahim Abe M. Elfadel},
  title     = {{PMIL: A Topology Module to Improve MIL-Based WSI Classification}},
  booktitle = {2025 IEEE International Symposium on Circuits and Systems (ISCAS)},
  year      = {2025},
  pages     = {1--5},
  publisher = {IEEE},
  address   = {London, United Kingdom},
  doi       = {10.1109/ISCAS56072.2025.11043738},
  url       = {https://doi.org/10.1109/ISCAS56072.2025.11043738},
}

@book{Oudot2015PersistenceT,
  author    = {Steve Y. Oudot},
  title     = {{Persistence Theory: From Quiver Representations to Data Analysis}},
  series    = {Mathematical Surveys and Monographs},
  volume    = {209},
  publisher = {American Mathematical Society},
  address   = {Providence, RI},
  year      = {2015},
  isbn      = {978-1-4704-2545-6},
  doi       = {10.1090/surv/209},
  url       = {https://doi.org/10.1090/surv/209},
}

@article{patel2018generalized,
  author  = {Amit Patel},
  title   = {{Generalized Persistence Diagrams}},
  journal = {Journal of Applied and Computational Topology},
  year    = {2018},
  volume  = {1},
  number  = {3--4},
  pages   = {397--419},
  doi     = {10.1007/s41468-018-0012-6},
  url     = {https://doi.org/10.1007/s41468-018-0012-6},
}

@article{pun2022survey,
  author  = {Chi Seng Pun and Si Xian Lee and Kelin Xia},
  title   = {{Persistent-Homology-Based Machine Learning: A Survey and a Comparative Study}},
  journal = {Artificial Intelligence Review},
  year    = {2022},
  volume  = {55},
  number  = {7},
  pages   = {5169--5213},
  doi     = {10.1007/s10462-022-10146-z},
  url     = {https://doi.org/10.1007/s10462-022-10146-z},
}

@article{WangCheng2012,
  author  = {Jia Wang and James Cheng},
  title   = {{Truss Decomposition in Massive Networks}},
  journal = {Proceedings of the VLDB Endowment},
  year    = {2012},
  volume  = {5},
  number  = {9},
  pages   = {812--823},
  month   = may,
  issn    = {2150-8097},
  doi     = {10.14778/2311906.2311909},
  url     = {https://doi.org/10.14778/2311906.2311909},
}

@article{Zomorodian2005ComputingPH,
  author  = {Afra Zomorodian and Gunnar Carlsson},
  title   = {{Computing Persistent Homology}},
  journal = {Discrete \& Computational Geometry},
  year    = {2005},
  volume  = {33},
  number  = {2},
  pages   = {249--274},
  doi     = {10.1007/s00454-004-1146-y},
  url     = {https://doi.org/10.1007/s00454-004-1146-y},
}

@article{kazeminia2026topomil,
  title   = {{TopoMIL}: Topology Improves Multiple Instance Learning in Diagnostic Microscopic Images},
  author  = {Kazeminia, Salome and Dasdelen, Muhammed Furkan and Rieck, Bastian and Marr, Carsten},
  journal = {bioRxiv},
  year    = {2026},
  doi     = {10.64898/2026.06.10.731443}
}

@misc{beuria2025probing,
  title         = {Probing Invisible Decay of {$Z^\prime$} at Muon Collider with Topological Data Analysis and Machine Learning},
  author        = {Beuria, Jyotiranjan},
  year          = {2025},
  eprint        = {2509.19756},
  archivePrefix = {arXiv},
  primaryClass  = {hep-ph},
  doi           = {10.48550/arXiv.2509.19756},
  url           = {https://arxiv.org/abs/2509.19756}
}

\begin{appendices}

\section{Implementation details}
\label{app:implementation-details}

\paragraph{Random network generation and persistence diagrams.}
For the random network model benchmark in Section~\ref{subsec:benchmark}, we generate graphs from the ten models with the following parameter conditions. We fix root seed \(20260601\) and derive all stochastic generation from fixed experiment seeds. For the Erd\H{o}s--R\'enyi condition, we sample \(200\) dyads uniformly without replacement from the \(\binom{100}{2}\) possible dyads. For the Watts--Strogatz conditions, we use degree \(4\) and rewiring probability in \(\{0.05,0.10,0.20,0.40,0.80\}\). For the Barab\'asi--Albert condition, we use an initial clique of size \(5\) and attachment count \(2\). For the Configuration Model condition, we sample a \(4\)-regular configuration and repeat the pairing until we obtain a simple graph.

For the Stochastic Block Model conditions, we use four blocks and log-odds contrast in \(\{0.5,1.0,1.5,2.0,3.0\}\). We randomly permute the vertices before assigning them as evenly as possible among the four blocks. For each contrast, we calibrate a common logit intercept so that the corresponding independent Bernoulli model has expected edge count \(200\), and we then sample the dyads conditional on exactly \(200\) edges. For the Chung--Lu conditions, we use degree exponent in \(\{2.2,2.4,2.6,2.8,3.0\}\), calibrate the power-law weight scale to expected edge count \(200\), and sample conditional on exactly \(200\) edges.

For the Kleinberg Small-World conditions, we use a \(10\times10\) grid, add \(20\) long-range edges without replacement, and use distance exponent in \(\{1.0,1.5,2.0,2.5,3.0\}\). We weight each admissible long-range candidate by its Manhattan distance raised to the negative distance exponent. For the GIRG conditions, the two varying parameters are the degree exponent and connection exponent, and we use \((2.2,1.0)\), \((2.4,1.0)\), \((2.6,1.0)\), \((2.6,1.5)\), and \((2.6,2.0)\). We sample one-dimensional torus positions uniformly, calibrate the connection scale to expected edge count \(200\), and sample conditional on exactly \(200\) edges. For the hyperbolic random graph conditions, the parameter coordinates are the degree exponent, temperature, curvature, and radius. We use \((2.2,0.3,-1.0,8.0)\), \((2.4,0.3,-1.0,8.0)\), \((2.6,0.3,-1.0,8.0)\), \((2.6,0.5,-1.0,8.0)\), and \((2.6,0.7,-1.0,8.0)\), and we sample dyads from the corresponding hyperbolic connection probabilities conditional on exactly \(200\) edges.

For the Exponential Random Graph Model conditions, we use a fixed-edge-count Markov chain with GWESP coefficient in \(\{0.25,0.50,0.75,1.00,1.25\}\), GWESP decay \(0.5\), \(100\) burn-in sweeps, and \(20\) additional sampling sweeps. Each proposal exchanges one present dyad with one absent dyad, so every state contains \(200\) edges.

Each generator assigns every selected edge a model-specific score and an order of formation. For the models based on dyad probabilities, this score equals the corresponding dyad probability. For the Barab\'asi--Albert generator, it equals the attachment probability at the time of selection. For the Watts--Strogatz generator, it equals the probability associated with retaining or creating the edge. For the Erd\H{o}s--R\'enyi and Configuration Model generators, the score is constant within a graph. For the Kleinberg generator, local edges have score \(1\), and each long-range edge receives its normalized selection weight at the time of selection. For the Exponential Random Graph Model, the final score is the GWESP-based deletion score used by the implementation. We order edges first by decreasing score, then by their order of formation, and finally lexicographically by their endpoint pair. If \(e_j\) is the \(j\)-th edge in this ordering for \(j=0,\ldots,m-1\), we assign filtration value \(j/(m-1)\). We compute the \(H_1\) persistence diagram of the resulting clique filtration.

\paragraph{Benchmark representations.}
For PL, we use three landscapes with resolution \(24\). We fit the landscape sampling range using only the persistence diagrams in the relevant training set, transform held-out diagrams with that fitted range, and assign the zero vector to an empty diagram. After inner model selection, we refit the PL transformation on the complete outer-training set before transforming the outer-test diagrams.

For PI, we use an \(8\times8\) grid and bandwidth \(0.10\). After transforming an interval \((b,d)\) into birth and persistence coordinates, we weight the resulting point by \(d-b\). We fit the birth and persistence ranges using only the relevant training diagrams. If either fitted coordinate range has zero width, we expand that range symmetrically by one bandwidth before refitting. We assign the zero vector to an empty diagram. For both PL and PI, we average the graph-level vectors coordinatewise within each batch.

For PersLay, we use a \(64\)-dimensional representation with one trainable landscape layer. For each training partition, we initialize the \(64\) landscape sample locations at equally spaced points between the minimum birth and maximum death among the persistence diagrams in the training partition; if all training diagrams are empty, we use the interval \([0,1]\). We use sum aggregation over persistence-diagram points and the PersLay power weighting function with exponent \(1\), with its trainable coefficient initialized to \(1\). For each batch, we compute a graph-level PersLay representation for each persistence diagram and average these representations coordinatewise within the batch. We train PersLay with a linear classification layer applied to the averaged batch representation and jointly optimize the PersLay parameters and the parameters of the linear classification layer using cross-entropy loss and Adam for \(25\) epochs with a learning rate of \(10^{-3}\). We use this linear layer only during PersLay training.

For HA, we use \(256\) complex-valued coordinates, represented by their real and imaginary parts as \(512\) real-valued features throughout the benchmark. We parameterize each sampled character by \(\omega=(\omega_1,\omega_2)\in\mathbb R^2\), whose coordinates we sample independently from the standard normal distribution with a seeded \texttt{std::mt19937\_64} generator. For an interval \(u=(b,d)\), we set \(\psi_\omega(u)=\omega_1b+\omega_2d\) and use the character \(\chi_\omega([u])=\exp(i\psi_\omega(u))\). The representation uses \(63\) order-\(1\) HA coordinates, \(63\) order-\(1\) HA coordinates that use each coefficient multiplied by the persistence of its interval, \(63\) order-\(2\) HA coordinates, and \(63\) order-\(2\) HA coordinates that use each coefficient multiplied by the persistence of its interval. The remaining four coordinates are the exact evaluations of these four constructions at the trivial character. Each order-\(2\) coordinate uses the same sampled character at both endpoints. We select the sampled characters for each of the four \(63\)-coordinate constructions evenly across the fixed sampled character pool. For each batch, we compute the complete HA representation separately for each persistence diagram and average the resulting representations coordinatewise. We evaluate the order-\(2\) coordinates using Theorem~\ref{thm:harmonic-aggregation-factorization} without explicitly constructing the aggregate.

\paragraph{Nested cross-validation and preprocessing.}
We use the nested cross-validation protocol from Section~\ref{subsec:benchmark}. The fixed outer and inner fold seeds are \(6849471059393190731\) and \(12053211736635842437\), respectively. The RBF kernel uses the library setting \(\texttt{gamma}=\texttt{scale}\), and we use no class weighting. We select \(C\in\{10^{-2},10^{-1},1,10,10^2\}\) by inner-validation accuracy.

For every training set, we remove coordinates with zero variance, compute the mean and population standard deviation of each retained coordinate from that training set, and use those quantities to standardize both the training and held-out features. For each inner fold, we fit PL and PI and train PersLay using only the inner training set, and we use the resulting representations to select \(C\) on the inner validation set. After selecting \(C\), we refit PL and PI and retrain PersLay using the complete outer training set, then compute the outer-training and outer-test representations.

Each method yields one out-of-fold prediction for every benchmark batch. We compute pooled accuracy and macro-\(F_1\) from these predictions. For each individual method, we compute \(95\%\) percentile confidence intervals from \(10{,}000\) class-stratified bootstrap resamples. For paired comparisons, each bootstrap replicate uses the same resampled observations for both prediction vectors. We use benchmark bootstrap seed \(1524974878092561023\).

\paragraph{Ablation study.}
The controlled variants reuse the benchmark observations, folds, preprocessing procedure, classifier family, and \(C\)-grid.

For the cross-observation condition in Equation~\ref{eq:cross-observation-quadratic}, we pair the four persistence diagrams in each batch by a fixed-point-free permutation. We evaluate ten separately seeded fixed-point-free permutations, and each realization passes independently through the common nested cross-validation procedure. For the no-preorder condition in Equation~\ref{eq:no-preorder-quadratic}, we evaluate one realization through the same nested cross-validation procedure. For the linear condition, we use \(127\) order-\(1\) HA coordinates, \(127\) order-\(1\) HA coordinates that use each coefficient multiplied by the persistence of its interval, and the corresponding two exact evaluations at the trivial character, and we evaluate one realization through the same nested cross-validation procedure.

For the linear and no-preorder conditions, we compare the out-of-fold predictions with benchmark HA by a paired stratified bootstrap. For the cross-observation condition, we apply a hierarchical paired stratified bootstrap to the ten realization-specific prediction vectors and the fixed benchmark HA prediction vector. We use ablation bootstrap seed \(5417834325182179367\).

\paragraph{Batch-size sensitivity.}
For the batch-size sensitivity analysis, we evaluate \(b\in\{1,2,3,4,5,6,7\}\) on the same \(1000\) model, parameter-condition, and batch identities. Every value of \(b\) therefore uses the same labels and the same outer and inner folds. We retain the \(256\) complex-valued HA coordinates, represented as \(512\) real-valued features, and select \(C\) separately for each value of \(b\). We use \(b=4\) as the reference and compare every other out-of-fold prediction vector with the \(b=4\) prediction vector by the paired stratified bootstrap. We use batch-size bootstrap seed \(9173827419318802191\).

\paragraph{HA-coordinate sensitivity.}
For the HA-coordinate sensitivity analysis, we evaluate \(K\in\{64,128,256,512,1024\}\) with ten replicates. Within each replicate, we use one fixed sampled character pool for all five values of \(K\), and distinct replicates use independently sampled character pools. For \(K>4\), we divide the \(K-4\) sampled coordinates as evenly as possible among the four HA constructions and use the remaining four coordinates for their exact evaluations at the trivial character.

For every \(K\) and replicate, we reuse the benchmark observations and nested cross-validation folds and select \(C\) separately. The benchmark uses one fixed sampled character pool for its \(256\) complex-valued HA coordinates, represented as \(512\) real-valued features, whereas the sensitivity analysis reports the mean over ten independently sampled replicates when \(K=256\). The benchmark and sensitivity replicates therefore use different random samples. For each \(K\), we apply a hierarchical paired stratified bootstrap to the ten replicate-specific prediction vectors and the fixed benchmark HA prediction vector. We use bootstrap seed \(12309768185157444137\) for the HA-coordinate sensitivity analysis.

\paragraph{Robustness.}
We use \(\varepsilon\in\{0.02,0.05,0.10,0.20\}\) for edge deletion, edge insertion, and degree-preserving rewiring. We use \(\varepsilon\in\{0.01,0.025,0.05,0.10\}\) for filtration noise. For each perturbation and severity, we generate ten realizations with separate random seeds. For edge deletion at severity \(\varepsilon\), we randomly permute the graph's edges and delete the first \(\lfloor\varepsilon m\rfloor\) edges. For edge insertion at severity \(\varepsilon\), we randomly permute the graph's nonedges and add up to \(\lfloor\varepsilon m\rfloor\) edges from this ordering. We assign a filtration time to each inserted edge using a random permutation of the filtration times of the existing edges. We then recompute the filtration steps from the updated filtration times.

For degree-preserving rewiring, each requested operation selects uniformly from admissible double-edge swaps on two disjoint edges that preserve all vertex degrees and graph simplicity. The two replacement edges inherit the filtration times associated with the two replaced edges. We perform up to \(\lfloor\varepsilon m\rfloor\) such swaps and stop when the graph has no admissible swap.

For filtration noise, we sample an independent \(Z\sim N(0,1)\) for each edge and replace its filtration time \(t\) by \(\min\{1,\max\{0,t+\varepsilon Z\}\}\). We then order the perturbed times and recompute filtration steps, with lexicographic endpoint order resolving ties.

Every robustness realization uses the benchmark observations, labels, and nested cross-validation folds and selects \(C\) separately. For each perturbation and severity, we apply a hierarchical paired stratified bootstrap to the ten realization-specific prediction vectors and the fixed benchmark HA prediction vector. We use robustness bootstrap seed \(16114472213211239267\).

\paragraph{Real-world datasets.}
The real-world benchmark uses AIDS, BZR, COX2, MUTAG, Mutagenicity, NCI1, NCI109, and PROTEINS \cite{Morris+2020}. We compare HA, PI, PL, and PersLay using the same \(H_1\) persistence diagrams.

For each dataset and bag size \(b\in\{1,2,3,4,5,6,7\}\), we split graphs before forming bags, group graphs only with other graphs carrying the same class label, form bags separately within each split, use the same bags for all methods at a fixed \(b\), and discard incomplete remainders. We use bag-construction seed \(11390238479102561291\). Each condition uses five outer and three inner folds, with the same numerical outer and inner seed values specified above. Different bag sizes produce different bag collections and sample counts.

We use edge trussness as in \cite[p.~813]{WangCheng2012}. For an edge \(e=\{u,v\}\), we define
\[
q(e)
:=
\frac{|N(u)\cap N(v)|}
{\min\{\deg(u)-1,\deg(v)-1\}},
\]
with denominator \(1\) when the minimum in the denominator is zero. We order edges first by decreasing trussness and then by decreasing \(q\). Edges with equal values of both quantities enter at the same filtration level. If the filtration contains \(L\) levels and an edge enters at level \(j\), indexed from \(j=0\), we assign filtration value \((j+1)/(L+1)\). We compute \(H_1\) persistence from the two-skeleton of the resulting clique filtration.

\paragraph{Computational environment.}
We compiled the C++ components in Release configuration with Visual Studio 2022 Build Tools and CMake 4.2.1 and ran the analysis with Python 3.13.2 on Windows 11 Pro. The workstation uses an Intel Core i7-14700KF processor and \(32\)~GB of RAM. The reported random network model and real-world data experiments run on the CPU, except that PersLay training and representation computation use CUDA on the GPU.

\section{Deferred proofs}\label{app:deferred-proofs}

\paragraph{Proof of Lemma~\ref{lem:higher-order-metric}.}
\label{app:proof-higher-order-metric}

\begin{proof}
We prove Lemma~\ref{lem:higher-order-metric} together with the closedness of \(\sim^{(n)}\) by simultaneous induction. Closedness of \(\sim^{(n)}\) implies that \(A^{(n+1)}\) is closed, which gives positive distance from \(A^{(n+1)}\) for every point outside it. At order \(0\), \(d^{(0)}=d\) is a metric on \(X^{(0)}=X\), and \(\sim^{(0)}\) is closed by assumption.

Fix \(n\geq0\). For \(n=0\), the preceding base case applies. For \(n\geq1\), assume that \(d^{(n)}\) is a metric on \(X^{(n)}\setminus A^{(n)}\) and that \(\sim^{(n)}\) is closed there. Let \(\rho\) denote the product \(p\)-metric induced by \(d^{(n)}\) on \(X^{(n+1)}\). By Definition~\ref{def:higher-order-persistence-intervals}, \(A^{(n+1)}=\{(a,b)\in X^{(n+1)}:a\sim^{(n)}b\}\). Closedness of \(\sim^{(n)}\) therefore implies that \(A^{(n+1)}\) is closed in \(X^{(n+1)}\). Hence every \(u\in X^{(n+1)}\setminus A^{(n+1)}\) has positive distance from \(A^{(n+1)}\).

Nonnegativity, symmetry, and \(d^{(n+1)}(u,u)=0\) follow from Definition~\ref{def:higher-order-persistence-intervals}. If \(u,v\in X^{(n+1)}\setminus A^{(n+1)}\), then \(d^{(n+1)}(u,A^{(n+1)})+d^{(n+1)}(v,A^{(n+1)})>0\). Thus \(d^{(n+1)}(u,v)=0\) implies \(\rho(u,v)=0\), and since \(\rho\) is a metric, \(u=v\).

We next prove the triangle inequality. For \(u,v\in X^{(n+1)}\) and \(x\in A^{(n+1)}\), the triangle inequality for \(\rho\) gives \(\rho(u,x)\leq\rho(u,v)+\rho(v,x)\). We take the infimum over \(x\in A^{(n+1)}\) to obtain \(d^{(n+1)}(u,A^{(n+1)})\leq\rho(u,v)+d^{(n+1)}(v,A^{(n+1)})\). Let \(u,v,w\in X^{(n+1)}\setminus A^{(n+1)}\). Definition~\ref{def:higher-order-persistence-intervals} gives \(d^{(n+1)}(u,w)\leq d^{(n+1)}(u,A^{(n+1)})+d^{(n+1)}(w,A^{(n+1)})\).

If \(d^{(n+1)}(u,v)=\rho(u,v)\) and \(d^{(n+1)}(v,w)=\rho(v,w)\), the triangle inequality for \(\rho\) gives the result. If \(d^{(n+1)}(u,v)=\rho(u,v)\) and \(d^{(n+1)}(v,w)=d^{(n+1)}(v,A^{(n+1)})+d^{(n+1)}(w,A^{(n+1)})\), then
\[
\begin{aligned}
d^{(n+1)}(u,A^{(n+1)})+d^{(n+1)}(w,A^{(n+1)})
&\leq \rho(u,v)+d^{(n+1)}(v,A^{(n+1)})
+d^{(n+1)}(w,A^{(n+1)}) \\
&=d^{(n+1)}(u,v)+d^{(n+1)}(v,w).
\end{aligned}
\]
The other mixed case follows symmetrically. Finally, suppose that \(d^{(n+1)}(u,v)=d^{(n+1)}(u,A^{(n+1)})+d^{(n+1)}(v,A^{(n+1)})\) and \(d^{(n+1)}(v,w)=d^{(n+1)}(v,A^{(n+1)})+d^{(n+1)}(w,A^{(n+1)})\). Then
\[
\begin{aligned}
d^{(n+1)}(u,A^{(n+1)})+d^{(n+1)}(w,A^{(n+1)})
&\leq d^{(n+1)}(u,A^{(n+1)})
+2d^{(n+1)}(v,A^{(n+1)}) \\
&\quad+d^{(n+1)}(w,A^{(n+1)}) \\
&=d^{(n+1)}(u,v)+d^{(n+1)}(v,w).
\end{aligned}
\]
Thus \(d^{(n+1)}\) satisfies the triangle inequality and is a metric on \(X^{(n+1)}\setminus A^{(n+1)}\).

We now prove that \(\sim^{(n+1)}\) is closed. Let \(u_k,v_k,u,v\in X^{(n+1)}\setminus A^{(n+1)}\) satisfy \(u_k\sim^{(n+1)}v_k\), \(u_k\to u\), and \(v_k\to v\). The second quantity in the minimum defining \(d^{(n+1)}(u_k,u)\) is at least \(d^{(n+1)}(u,A^{(n+1)})>0\). Since \(d^{(n+1)}(u_k,u)\to0\), we have \(d^{(n+1)}(u_k,u)=\rho(u_k,u)\) for all sufficiently large \(k\), and hence \(\rho(u_k,u)\to0\). Similarly, \(\rho(v_k,v)\to0\).

Write \(u_k=(a_k,b_k)\), \(v_k=(c_k,d_k)\), \(u=(a,b)\), and \(v=(c,d)\). Convergence in \(\rho\) gives convergence of each coordinate under \(d^{(n)}\). By Definition~\ref{def:higher-order-persistence-intervals}, \(u_k\sim^{(n+1)}v_k\) implies \(a_k\sim^{(n)}c_k\) and \(b_k\sim^{(n)}d_k\). Closedness of \(\sim^{(n)}\) gives \(a\sim^{(n)}c\) and \(b\sim^{(n)}d\), so \(u\sim^{(n+1)}v\).

The simultaneous induction proves Lemma~\ref{lem:higher-order-metric}.
\end{proof}

\paragraph{Proof of Proposition~\ref{prop:uniform-discreteness-levels}.}
\label{app:proof-uniform-discreteness-levels}

\begin{proof}
Let \(\varepsilon>0\) be the minimum distance between distinct classes in \(X^{(n)}/A^{(n)}\). Let \(u=(a,b)\) represent a non-basepoint class in \(X^{(n+1)}/A^{(n+1)}\), and let \((x,y)\in A^{(n+1)}\). By Definition~\ref{def:higher-order-persistence-intervals}, \(a,b,x,y\in X^{(n)}\setminus A^{(n)}\). Either \([a]\neq[x]\) or \([b]\neq[y]\). Otherwise, \([a]=[x]\) and \([b]=[y]\), and each non-basepoint class in \(X^{(n)}/A^{(n)}\) has a unique representative, so \(a=x\) and \(b=y\). Thus \(u=(x,y)\in A^{(n+1)}\), which contradicts that \(u\) represents a non-basepoint class. Uniform discreteness at level \(n\) therefore gives \(\|(d^{(n)}(a,x),d^{(n)}(b,y))\|_p\geq\varepsilon\). Hence \(d^{(n+1)}(u,A^{(n+1)})\geq\varepsilon\), so every non-basepoint class in \(X^{(n+1)}/A^{(n+1)}\) has distance at least \(\varepsilon\) from the basepoint.

Let \(u=(a,b)\) and \(v=(c,d)\) represent distinct non-basepoint classes in \(X^{(n+1)}/A^{(n+1)}\). By Definition~\ref{def:higher-order-persistence-intervals}, \(a,b,c,d\in X^{(n)}\setminus A^{(n)}\). Either \([a]\neq[c]\) or \([b]\neq[d]\). Otherwise, \([a]=[c]\) and \([b]=[d]\), and uniqueness of non-basepoint representatives gives \(a=c\) and \(b=d\), so \(u=v\), a contradiction. Uniform discreteness at level \(n\) therefore gives \(\|(d^{(n)}(a,c),d^{(n)}(b,d))\|_p\geq\varepsilon\). Since \(u\) and \(v\) represent non-basepoint classes, \(d^{(n+1)}(u,A^{(n+1)})+d^{(n+1)}(v,A^{(n+1)})\geq2\varepsilon\). Both quantities in the minimum defining \(d^{(n+1)}(u,v)\) are therefore at least \(\varepsilon\), so \(d^{(n+1)}(u,v)\geq\varepsilon\). Together with the basepoint bound, this shows that every pair of distinct classes in \(X^{(n+1)}/A^{(n+1)}\) has distance at least \(\varepsilon\), so \(X^{(n+1)}/A^{(n+1)}\) is uniformly discrete with minimum distance at least \(\varepsilon\).
\end{proof}

\paragraph{Proof of Lemma~\ref{lem:order-dimension-levels}.}
\label{app:proof-order-dimension-levels}

\begin{proof}
Let \(\phi_1,\ldots,\phi_{r_n}:X^{(n)}\to\mathbb R\) represent \(\preceq^{(n)}\). For \((a,b),(c,d)\in X^{(n+1)}\), Definition~\ref{def:higher-order-persistence-intervals} gives \((a,b)\preceq^{(n+1)}(c,d)\) if and only if \(c\preceq^{(n)}a\) and \(b\preceq^{(n)}d\), equivalently, \(-\phi_k(a)\leq-\phi_k(c)\) and \(\phi_k(b)\leq\phi_k(d)\) for every \(1\leq k\leq r_n\). Thus the \(2r_n\) maps \((a,b)\mapsto-\phi_k(a)\) and \((a,b)\mapsto\phi_k(b)\), for \(1\leq k\leq r_n\), represent \(\preceq^{(n+1)}\).
\end{proof}

\paragraph{Proof of Theorem~\ref{thm:implicit-harmonic-complexity}.}
\label{app:proof-implicit-harmonic-complexity}

\begin{proof}
Since \(\psi\) is constant on each equivalence class of \(\sim^{(n)}\), the terms with \(u\sim^{(n)}v\) vanish, so we may include them when evaluating the exponent in Equation~\ref{eq:aggregate-character}. We fix \(j\) and write the contribution of \(\Gamma_j\) to this exponent as
\begin{align*}
&
\sum_{v\in\operatorname{supp}(\Gamma_j)}
\Gamma_j(v)\psi(v)
\sum_{\substack{
u\in\operatorname{supp}(\Gamma_j)\\
u\preceq^{(n)}v
}}
\Gamma_j(u)
\\
&\quad-
\sum_{u\in\operatorname{supp}(\Gamma_j)}
\Gamma_j(u)\psi(u)
\sum_{\substack{
v\in\operatorname{supp}(\Gamma_j)\\
u\preceq^{(n)}v
}}
\Gamma_j(v).
\end{align*}
We evaluate the \(r_n\) maps representing \(\preceq^{(n)}\) on the support elements. The first family of inner sums consists of weighted dominance queries for the inequalities defining \(u\preceq^{(n)}v\). Negating the values of all \(r_n\) representing maps converts the second family into dominance queries of the same form. For fixed \(r_n\geq2\), we compute each family by multidimensional divide-and-conquer in \(O(|\operatorname{supp}(\Gamma_j)|\log^{r_n-1}|\operatorname{supp}(\Gamma_j)|)\) time \cite{10.1145/358841.358850}. Replacing count addition by signed-coefficient addition leaves the number of arithmetic operations unchanged under the unit-cost RAM model. We process equal values of the representing maps together so that the computation includes equality in the inequalities representing \(\preceq^{(n)}\). For \(r_n=1\), we sort the support according to the representing map and compute the required prefix and suffix sums in \(O(|\operatorname{supp}(\Gamma_j)|\log|\operatorname{supp}(\Gamma_j)|)\) time. We sum these bounds over \(1\leq j\leq m\) to obtain Equation~\ref{eq:implicit-character-complexity}.

The dominance sums depend only on the input coefficients and the values of the maps representing \(\preceq^{(n)}\). We compute these sums once and reuse them for the \(K\) functions that define the characters in Equation~\ref{eq:phase-character}. For each function, the two outer sums in the displayed decomposition require \(O(\sum_{j=1}^{m}|\operatorname{supp}(\Gamma_j)|)\) time. We add the \(K\) linear-time evaluations to the shared dominance computation to obtain Equation~\ref{eq:multiple-character-complexity}.
\end{proof}

\paragraph{Proof of Theorem~\ref{thm:harmonic-aggregation-factorization}.}
\label{app:proof-harmonic-aggregation-factorization}

\begin{proof}
Definitions~\ref{def:harmonic-aggregation} and~\ref{def:aggregation} give
\[
\begin{aligned}
&\operatorname{HA}_{n+1}\!\left(
\mathcal A_n(\Gamma_1,\ldots,\Gamma_m);
\chi_{\psi_L},\chi_{\psi_R}
\right)\\
&\qquad=
\frac{1}{2}
\sum_{j=1}^{m}
\sum_{\substack{
u,v\in\operatorname{supp}(\Gamma_j)\\
u\preceq^{(n)}v,\;
u\not\sim^{(n)}v
}}
\Gamma_j(u)\Gamma_j(v)
\left[
\chi_{\psi_L}([(u,v)])
+
\chi_{\psi_R}([(u,v)])
\right].
\end{aligned}
\]
By Equation~\ref{eq:phase-character},
\[
\chi_{\psi_L}([(u,v)])
=
e^{i\psi_L(v)}e^{-i\psi_L(u)}
\qquad\text{and}\qquad
\chi_{\psi_R}([(u,v)])
=
e^{i\psi_R(v)}e^{-i\psi_R(u)}.
\]
Substituting these identities and collecting each finite sum by \(v\) gives Equation~\ref{eq:harmonic-aggregation-factorization}.
\end{proof}

\paragraph{Proof of Theorem~\ref{thm:harmonic-aggregation-stability}.}
\label{app:proof-harmonic-aggregation-stability}

\begin{proof}
Fix \(n\geq1\) and \(\chi_L,\chi_R\in\widehat{K^{(n)}(X,A)}\). For distinct non-basepoint classes \([u],[v]\in X^{(n)}/A^{(n)}\), the triangle inequality and \(|\chi_L|=|\chi_R|=1\) give
\[
\left|
\frac{\chi_L([u])+\chi_R([u])}{2}
-
\frac{\chi_L([v])+\chi_R([v])}{2}
\right|
\leq 2
\leq
\frac{2}{\varepsilon}d^{(n)}(u,v).
\]
For a non-basepoint class \([u]\), the corresponding value has modulus at most \(1\), while Proposition~\ref{prop:uniform-discreteness-levels} gives \(d^{(n)}(u,A^{(n)})\geq\varepsilon\). Thus the assignment that sends the basepoint to \(0\) and each non-basepoint class \([u]\) to \((\chi_L([u])+\chi_R([u]))/2\) is \(2/\varepsilon\)-Lipschitz. By \cite[Theorem~7.4]{bubenik2022virtual}, this assignment extends uniquely to a linear map on \(V^{(n)}(X,A)\) with the same Lipschitz constant. Definition~\ref{def:harmonic-aggregation} identifies this linear extension with \(\operatorname{HA}_n(\,\cdot\,;\chi_L,\chi_R)\), which proves the result.
\end{proof}

\end{appendices}

\end{document}